\let\originalcomma=\,
\let\originalcolon=\:
\let\originalsemicolon=\;

\documentclass[envcountsect, envcountsame]{llncs}

\newif\ifconf
\conftrue

\ifconf\else
\usepackage[margin=1.4in]{geometry}
\fi

\title{Adding an Abstraction Barrier to ZF Set Theory}
\author{Ciar\'an Dunne \and J.~B.~Wells \and Fairouz Kamareddine}
\institute{Heriot-Watt University}

\usepackage{enumitem}
\usepackage
  [disable]
  {todonotes}

\usepackage{colonequals}
\usepackage{lmodern}
\usepackage{amsmath}
\usepackage{amsfonts}
\usepackage{tipa}
\usepackage{proof}
\usepackage{hyperref}
\usepackage{datetime}
\usepackage{jbw-fix-hyperref-autoref}

\let\,=\originalcomma
\let\:=\originalcolon
\let\;=\originalsemicolon

\newcommand{\jbwshortnote}[1]{\todo[color=red!40]{Joe: #1}}
\newcommand{\jbwlongnote}[1]{\todo[inline, color=red!40]{Joe: #1}}
\newcommand{\cmdshortnote}[1]{\todo[color=blue!40 ]{Ciaran: #1}}
\newcommand{\cmdlongnote}[1]{\todo[inline, color=blue!40]{Ciaran: #1}}

\newcommand{\var}[1]{\mathsf{v#1}}
\newcommand{\cceq}{\coloncolonequals}
\newcommand{\elem}{\mathrel{\dot{\in}}}
\newcommand{\abbreviates}{\mathrel{\colonequals}}
\DeclareRobustCommand{\riota}{\rotatebox[origin=C]{180}{$\iota$}}
\newcommand{\bind}[3]{#1\,#2:#3}
\newcommand{\defdes}[2]{\bind\riota{#1}{#2}}
\newcommand{\Exists}[2]{\bind\exists{#1}{#2}}
\newcommand{\exOneSym}{\exists!}
\newcommand{\exOne}[2]{\bind\exOneSym{#1}{#2}}
\newcommand{\lang}{\mathsf{Form}}
\newcommand{\Var}{\mathsf{Var}}
\newcommand{\zfp}{ZFP}
\newcommand{\zf}{ZF}
\newcommand{\zin}{\mathrel{\widehat{\in}}}

\newcommand{\pl}{\mathrel{\pi_1}}
\newcommand{\pr}{\mathrel{\pi_2}}

\newcommand{\allS}{\forall_{\mathsf{Set}}\,}
\newcommand{\allP}{\forall_{\mathsf{Pair}}\,}

\newcommand{\exS}{\exists_{\mathsf{Set}}\,}
\newcommand{\exP}{\exists_{\mathsf{Pair}}\,}

\newcommand{\unS}{\riota_{\mathsf{Set}}\,}
\newcommand{\unP}{\riota_{\mathsf{Pair}}\,}
\newcommand{\set}[1]{\{#1\}}
\newcommand{\AllSuchThat}[2]{\{\,#1\mid#2\,\}}
\newcommand{\pair}[2]{\langle #1,#2\rangle}
\newcommand{\Ord}{\mathsf{Ord}}

\newcommand{\Set}{\mathsf{Set}}
\newcommand{\Pair}{\mathsf{Pair}}

\newcommand{\zpl}{\mathrel{\widehat{\pi}_1}}
\newcommand{\zpr}{\mathrel{\widehat{\pi}_2}}
\newcommand{\pow}{\mathcal{P}}
\newcommand{\Union}{\cup\,}
\renewcommand{\succ}[1]{{#1}^+}
\newcommand{\W}{\mbox{\textbf{W}}}
\newcommand{\V}{\textbf{V}}

\spnewtheorem*{nota}{Notation}{\bfseries}{\normalfont}

\usepackage{navigator}
\embeddedfile[Isabelle Proof of Section 4]{proof}{ZFP.thy}

\begin{document}
\makeatletter
\def\toclevel@author{999}%
\@ifundefined{keywords}
  {\def\keywords#1{KEYWORDS NOT IMPLEMENTED.  USE A MORE RECENT llncs.cls.}}
  {}%
\makeatother

\maketitle
\begin{abstract}
\parindent=10pt\relax
Much mathematical writing exists that is,
explicitly or implicitly, based on set theory,
often Zermelo-Fraenkel set theory (\zf)
or one of its variants.
In {\zf}, the domain of discourse contains only sets, and
hence every mathematical object must be a set.
Consequently, in {\zf} with the usual encoding of an ordered pair $\pair{a}{b}$, formulas like
$\{a\} \in \pair{a}{b}$ have truth values, and operations like $\pow{(\pair{a}{b})}$ have results
that are sets.
Such `accidental theorems' do not match how people think about the mathematics and also cause practical difficulties when using set theory in machine-assisted theorem proving.
In contrast, in a number of proof assistants, mathematical
objects and concepts can be built of type-theoretic stuff so that
many mathematical objects can be, in essence, terms of an extended typed $\lambda$-calculus%
.
However, dilemmas and frustration arise when formalizing mathematics in type theory.

Motivated by problems of formalizing mathematics with (1)~purely
set-theoretic and (2)~type-theoretic approaches, we explore an option
with much of the flexibility of set theory and some of the useful
features of type theory.
We present {\zfp}: a modification of {\zf} that has ordered pairs as primitive, non-set objects.
{\zfp} has a more natural and abstract axiomatic definition of ordered pairs free of any notion of representation.
This paper presents axioms for {\zfp}, and a proof in {\zf}
(machine-checked in Isabelle/ZF) of the existence of a model for
{\zfp}, which
implies that {\zfp} is consistent if {\zf} is.
We discuss the approach used to add this abstraction barrier to {\zf}%
.
\end{abstract}

\section{Introduction}\label{sec:introduction}

\subsection{Background: Set Theory  and Type Theory as Foundations}
A large portion of the mathematical literature is based on set theory, explicitly or implicitly,
directly or indirectly.
Set theory is pervasive in mathematical culture.
University mathematics programmes have introductory courses on set
theory and many other courses that rely heavily on set-theoretic concepts (sets,
classes, etc.), notation (comprehensions a.k.a.\ set-builders, power
set, etc.), and reasoning.

Formal foundations for mathematics have been developed since the early
20th century, with both set-theoretic and type-theoretic approaches
being considered.
Although there are a number of set-theoretic foundations, for this
paper it is sufficient to consider Zermelo-Fraenkel set theory (\zf),
which anyway seems to be broadly accepted and reasonably representative of the strengths and
weaknesses of set theory in actual practice.
The core concept of {\zf} is the \emph{set membership relation} $\in$, which
acts on a domain of objects called \emph{sets}.
The theory is a collection of formulas (known as \emph{axioms}) of
first-order logic which characterise the membership relation.
Logical deduction from these axioms
yields a rich theory of sets.
Moreover, mathematical objects such as ordered pairs,
functions, and numbers can be represented as sets in {\zf}.
%

At roughly the same time as Zermelo was formulating his axiomatic set
theory, Russell introduced the first type theory.
Both Zermelo and Russell had the goal of rigorous, formal, logical
reasoning free from the paradoxes that plagued the earlier systems of
Cantor and Frege.
Most modern type theories are descendants of Church's typed
$\lambda$-calculus~\cite{kubota}.
Many of the methods of modern type theory have been developed by
computer scientists to solve problems in programming
languages and formal verification.
Types add layers of reasoning that help with soundness and representation independence.
Some type theories have been used to formulate foundations of
mathematics in which mathematical objects (e.g., groups, rings, etc.)
are represented by terms and types of what is essentially a very fancy typed
$\lambda$-calculus.

\jbwlongnote{Maybe discuss the Curry-Howard correspondence.}
\jbwlongnote{Type-theoretic proof systems such as Coq, Agda, Lean.}
\jbwlongnote{Important: Compare with HOL!}
\cmdlongnote{More benefits of type theory?}

Formalizing mathematics that has been developed in a set-theoretic
culture using a type-theoretic foundation can lead to
dilemmas and frustration~\cite{harrison}.
Subtyping may not work smoothly when formalising chains of structures
such as the number systems and those belonging to universal
algebra.
There are also design choices in how to model predicates which can make proving some things easier
but other things much harder.
The rules of powerful type systems are also very complicated, so users
require machine assistance to follow the typing rules, and even with
machine support it can be quite challenging.
In contrast, ZF-like set theories typically have very few
`types', e.g., there might be a type of sets and a type of logical formulas or perhaps a type of
classes.
When nearly every mathematical object you need is of `type set' it is easy to obey the typing
rules.

There are problems formalizing mathematics in pure ZF set theory also.
When everything is of `type set', a computer proof system has
no easy way to know that it would be wasting its time to try to prove
a theorem about ordinal numbers using lemmas and tactics for groups or
rings, so automated support is more challenging.
When representing mathematical objects (e.g., numbers) as sets, the
bookkeeping of the
intended `type' of these objects is not avoided, but must be managed by the user outside the realm
of a type system.
In many not-too-tricky cases, a type inference algorithm can
automatically infer type information that represents necessary
preconditions for successful use of theorems and lemmas, but in pure set theory such
automated inference is not very useful when the only type is `set'.

Furthermore, practical computerisation in {\zf} requires abbreviation and definition mechanisms
which first-order logic does not provide.
Two contrasting examples of how this can be done are Metamath and Isabelle/ZF.
Metamath~\cite{metamath} is mostly string based, and has `syntax definitions' to introduce new constants, or syntax patterns.
These definitions are given meaning by `defining axioms' (whose correctness is not checked by the verifier).
Isabelle/ZF is built on top of Isabelle/Pure, which is a fragment of intuitionistic higher-order
logic that is based on Church's typed $\lambda$-calculus~\cite{paulson}.
This means that meta-level activities such as variable binding, definitions, and abbreviations are
handled by Isabelle/ZF in a type theory, albeit a very simple type theory.
Isabelle also handles proof tactics in SML, which can be seen as another typed $\lambda$-calculus.


%


\jbwlongnote{list computer proof systems that target ALL mathematics and (1)~use a set-theoretic
  foundation of mathematics, (2)~use a type-theoretic foundation of mathematics, or (3)~don't fit in
  either of the two previous categories.}


\subsection{The Issue of Representation and the Case of the Ordered Pair}\label{sec:issue-representation}


As discussed above, set theory can represent a multitude of mathematical objects as sets,
but in some cases the user might prefer that some of their mathematical objects are genuinely not
sets.
The alternative of using a sophisticated type-theoretic foundation might not be the right solution,
for a variety of reasons, some of which are mentioned above.
So the user might ask: ``May I please have a set theory which has genuine non-sets that I can use for
purpose XYZ?''

There are indeed set theories with non-set objects~\cite{holmes}, which are generally known as \emph{urelements},
so named because they are often considered to be primordial, existing independently of and before
the sets.
A popular use for urelements is as `atoms' whose only properties are being distinct from everything
else and existing in large enough multitudes.
Adding genuine non-sets takes some work, because the assumption that `everything is a set' is deeply
embedded in {\zf}'s axioms.
One example is the axiom of Extensionality,
$$
  \forall x,y: (\forall a: a\in x \leftrightarrow a \in y) \rightarrow x=y
$$
which asserts that any two objects are equal if they have exactly the same set members.
Because non-set objects of course have no set members, this {\zf} axiom forces them all to equal the
empty set, meaning there can not be any.

Existing set theories with urelements generally (except see GST below) do not consider urelements
with `internal' structure that might include sets.
The \emph{ordered pair} is a simple and important example of a mathematical object with `internal'
structure which is not usually intended to be viewed as a set.
Ordered pairs have been of enormous value in building theories of relations, functions, and spaces.
The most widely used set-theoretical definition, by Kuratowski, defines the ordered pair
$\pair{a}{b}$ to be the set $\{\{a\},\{a,b\}\}$.
Because $a$ is in all sets in $\pair{a}{b}$ and $b$ is only in one, a first-order logic formula
using only the membership relation can check if an object is the first (or second) projection of an ordered
pair.
Kuratowski pairs satisfy the characteristic property of ordered pairs:
$$
  \pair{a}{b} = \pair{c}{d} \leftrightarrow (a = c \wedge b = d)
$$
Like for any {\zf} representation of mathematical objects not thought of as sets, Kuratowski pairs have
`accidental theorems' such as $\set{b} \in \pair{b}{c}$, and
\(
   \set{\pair{b}{b}}
  =\set{\set{\set{b}}}
  =\pair{\set{b}}{\set{b}}
\),
and $\set{\pair{0}{0}}=\pair{1}{1}$ with Von Neumann numbers.

The set representation of conceptually non-set objects raises issues.
There are places in the literature where some mathematical objects are thought of as (or even
explicitly stated to be) non-sets with no set members.
One can find
definitions or proofs by cases on `type'
that assume the case of sets never overlaps with the cases of
pairs, numbers, etc.
To view such writing as being founded on pure set theory requires either proving that none of the
sets used overlap with the set representations used for abstract objects or inserting many tagging
and tag-checking operations (see, e.g., the translation we give in \autoref{def:translation} as part
of proving a model for our system {\zfp} can be built in the pure set theory {\zf}).
When formalizing and machine-checking mathematics, additional difficulties arise, some of which are
mentioned above.

\subsection{ZFP: Extending ZF Set Theory with Primitive Ordered Pairs}

We aim to go beyond previous set theories with urelements to develop methods for extending set
theories with genuine non-set objects whose internal structure can contain other objects including
the possibility of sets.
As a first instance of this aim, we achieve the objective of {\zfp}, a set theory with primitive
non-set ordered pairs such that there is no limit on the `types' of objects that sets and ordered
pairs may contain.
We axiomatise {\zfp} and prove its consistency relative to {\zf}.
We hope that our explanation of how we did this will be useful guidance for other work extending set
theories.

{\zfp} extends $\in$ with two new binary predicate symbols, $\pl$ and $\pr$, whose intended meanings
are `is first projection of' and `is second projection of'.
We define abbreviations for formulas $\Set(x)$ and $\Pair(x)$ that distinguish sets and ordered
pairs by the rule that an ordered pair has a first projection and a set does not.
{\zfp}'s axioms are in two groups, one for sets and one for ordered pairs.
We were able to generate nearly all of {\zfp}'s axioms for sets by modifying the axioms of {\zf}
by restricting quantifiers using $\Set(x)$ in the right places.
The axiom of Foundation needed to be modified to handle sets and ordered pairs
simultaneously\jbwlongnote{ and to implement our decision that there could not be infinite
  descending chains via ordered pairs which enabled building a model in a way similar to building models for
  ZF}.
{\zfp}'s axioms for ordered pairs specify the expected abstract properties, including that ordered
pairs have no set members.


To prove {\zfp} is consistent if {\zf} is, we construct in {\zf} a model and prove it satisfies
{\zfp}'s axioms~\cite{enderton}\jbwshortnote{purpose of citation unclear to reader}.
Building a model for a set theory with non-set objects with `internal' structure that can include
sets differs from building a model for a set theory with no urelements or with only simple
urelements, because there can be new non-set objects at each stage of the construction.
\W, the domain of our model, is similar to the domain {\V} of the Von Neumann hierarchy.
Each tier of {\V} is constructed by taking the power set of the previous tiers.
In contrast, when building the tiers of \W, each successor tier $W_{\alpha^+}$ is formed by taking the
disjoint sum of the power set $\pow(W_{\alpha})$ and the cartesian product $(W_\alpha)^2$.
Hence every object in {\W} has a tag that tells whether it is intended to model a set or an ordered
pair.
This supports defining relations that model ZFP's $\in$, $\pl$, and $\pr$ which may only return true
when their second argument is of the correct `type'.
This proof has been machine-checked in Isabelle/{\zf}.%
\footnote{See \url{http://www.macs.hw.ac.uk/~cmd1/cicm2020/ZFP.thy} for the source, and
  \url{http://www.macs.hw.ac.uk/~cmd1/cicm2020/ZFPDoc/index.html} for the HTML.}

Although our model for {\zfp} is built purely of sets and implements ordered pairs as sets,
another model could use other methods (e.g., type-theoretic) and
implement ordered pairs differently.
Hence, we have put an `abstraction barrier' between the user of {\zfp} and the implementation of
ordered pairs.

\subsection{Related Work}

Harrison~\cite{harrison}
details the challenges that face both type-theoretic and set-theoretic foundations for formalised
mathematics.
Harrison makes the case for using set theory as `machine code', leaving theorem proving to layers of
code\jbwshortnote{huh?}.
Harrison suggests using a set theory with urelements to avoid the issue of `accidental theorems'.
%
Weidijk~\cite{wiedijk}
formulates axiomatic set theories and type theories in AutoMath in order to compare them and assess
their relative complexity.

\jbwlongnote{Discuss Mizar!}

A significant work aiming to make computer formalisation of set-theoretical mathematics practical is
Farmer's Chiron~\cite{farmer}, a conservative extension of the set theory
NBG (itself a conservative extension of {\zf}).
Chiron has additional features such as support for
undefinedness, definite descriptions,
quotation and evaluation of expressions, and a kind of types.

Aczel and Lunnon worked on Generalised Set Theory (GST)~\cite{aczel} with the aim of better
supporting work in situation theory\jbwshortnote{weird to mention Lunnon but not cite anything by
  her}.
GST extends set theory with a mechanism for primitive functions, as well as a number of
other features.
It appears that GST assumes the Anti-Foundation axiom instead of Foundation which {\zf} uses.
Unfortunately, we failed to find a specification of the axioms of GST.
Part of GST seems similar to our work but a technical comparison is difficult without the
axioms.

Although ordered pairs now seem obvious, Kanamori's excellent history~\cite{kanamori}
shows a sequence of conceptual breakthroughs were needed to reach the modern ordered pair.
How we built a model for {\zfp} was heavily inspired by the way
Barwise~\cite{barwise} interprets KPU (Kripke-Platek set theory with Urelements) in KP.

\jbwlongnote{Holmes has a big book based I think on NFU in which ordered pairs are considered as
  urelements.
  We must cite and discuss this!}

\subsection{Outline}

\Autoref{sec:fol-zf} presents and discusses the first-order logic we use and definitions and axioms of {\zf}.
\Autoref{sec:zfp} presents and discusses {\zfp} in the form of definitions and two collections of axioms,
one for sets, and one for ordered pairs.
\Autoref{sec:zfp-model} proves the existence in {\zf} of a model for the axioms of {\zfp} (which implies that
{\zfp} is consistent if {\zf} is)%
.
\Autoref{sec:future} discusses the significance of these results, and how they will be used in further
investigation.

\section{Formal Machinery}\label{sec:fol-zf}


Let $X\abbreviates Y$ be meta-level notation meaning that $X$ stands for $Y$.

\subsection{First-Order Logic with Equality}

We use a fragment of first-order logic (FOL) with equality sufficient for defining {\zf} and {\zfp}.
We consider only four binary infix \emph{predicate symbols} including equality.
The MBNF~\cite{mbnf} specification of the syntax is:
$$%
  \begin{array}{@{}l@{\qquad}l@{}}
      a,\ldots,z   \elem \Var          \cceq\var{0} \mid \var{1} \mid \cdots
    &
      {\sim} \elem \mathsf{Pred} \cceq {\in} \mid \pi_1 \mid \pi_2 \mid {=}
  \\
      A, \ldots, Z \elem \mathsf{Term} \cceq x \mid \riota x: \varphi
    &
      \varphi,\psi \elem \lang         \cceq X\mathbin{\sim}Y \mid \varphi\rightarrow\psi\mid\neg\varphi\mid\forall x:\varphi
  \end{array}
$$%
We work with terms and formulas modulo $\alpha$-conversion where $\forall x$ and $\riota x$ bind
$x$.
Except where explicitly specified otherwise, we require metavariables ranging over the set $\Var$ to
have the attribute of \emph{distinctness}.
Two different metavariables with the distinctness attribute can not be equal.
For example, $x=\var{9}$ and $x_1=\var{27}$ and $y=\var{53}$ could hold simultaneously, but neither
$x=\var{9}=x_1$ nor $x=\var{53}=y$ are allowed.
This restriction applies only to metavariables: the same object-level variable can be used in nested
scopes, e.g., the formula
$(\bind{\forall}{\var{7}}{\bind{\forall}{\var{7}}{\var{7}\in\var{7}}})$ is fine and equal to
$(\bind{\forall}{\var{0}}{\bind{\forall}{\var{1}}{\var{1}\in\var{1}}})$.
We assume the usual abbreviations for logical connectives ($\wedge$, $\vee$, $\leftrightarrow$),
for quantifiers ($\exists$, $\exOneSym$, $\forall x_1,\ldots, x_n$, $\exists x_1, \ldots, x_n$), and
for predicate symbols ($\ne$, $\notin$, $\ni$).
\jbwlongnote{Specify precedence and associativity, including for operators defined by
  abbreviations.}

A term can be a \emph{definite description} $(\defdes x\varphi)$ which, if there is
exactly one member $x$ of the domain of discourse such that the formula $\varphi$ is true, evaluates
to that member and otherwise evaluates to a special value $\bot$ outside the domain of discourse such
that any predicate symbol (including equality) with $\bot$ as an argument evaluates to false.%
\footnote{When working with functions that might be applied outside their domain, one might prefer to
  have $\bot=\bot$, but this is a bit more complex and not needed for this paper.}
\relax
A term is said to be \emph{undefined} or to \emph{have no value} iff it evaluates to $\bot$.
An alternative specification of definite descriptions that gives formulas the same meanings
is eliminating them by
the following rule (only the left case is given; the
right case is similar):
$$
    ((\defdes x\varphi)\mathbin{\sim} Y)
  \abbreviates
      (\Exists{x}
         {       x \mathbin{\sim} Y
          \wedge \varphi})
    \wedge
      \exOne{x}{\varphi}
  \mbox{ where $x$ is not free in $Y$}
$$

\subsection{Zermelo-Fraenkel Set Theory}

The only predicate symbols {\zf} uses are the \emph{membership relation} $\in$ and equality.
{\zf} makes no use of the FOL predicate symbols $\pl$ and $\pr$, but instead we define these symbols
as parts of abbreviations in~\autoref{sec:ordered-pairs-zf}.
We use the following abbreviations
where $n\geq 3$ and $a$, $c$, $x$, $y$, and $z$ are not free in the other arguments and $b$ is not
free in $X$:
$$%
  \begin{array}{@{}l@{\quad}l@{}}
      (\forall b\in X:\varphi)   \abbreviates (\forall b: b\in X \rightarrow \varphi)
    &
      (\exists b \in X: \varphi) \abbreviates (\exists b: b\in X \wedge \varphi)
  \\
      \Union X \abbreviates (\riota y: \forall a: a\in y \leftrightarrow \exists z\in X: a\in z)
    &
      X \subseteq Y \abbreviates (\bind{\forall}{c\in X}{c\in Y})
  \\
      \{A,B\} \abbreviates (\bind{\riota}{x}{\forall c: c \in x \leftrightarrow (c = A \vee c = B)})
    &
      X \cup Y \abbreviates \Union\{X,Y\}
  \\
      \pow(X) \abbreviates (\bind{\riota}{y}{\bind{\forall}{z}{z \in y \leftrightarrow z \subseteq X}})
    &
      \{A\} \abbreviates \{A,A\}
  \\
      \{A_1,\ldots, A_n\} \abbreviates \{A_1\} \cup \{A_2,\ldots, A_n\}
    &
      \emptyset \abbreviates (\riota x: \forall a: a\notin x)
  \\
      \multicolumn{2}{@{}l@{}}
        {               \AllSuchThat{b \in X}{\varphi}
           \abbreviates (\bind{\riota}{y}
                           {\forall b: b \in y \leftrightarrow (b \in X \wedge \varphi)})
         \quad
           \succ{X} \abbreviates X \cup \set{X}
        }
  \end{array}
$$
These abbreviations are defined if their arguments are defined due to the axioms.
\jbwlongnote{Warning: These are essentially macros, not predicate symbols or function symbols
  specified by axioms.
  For brevity they don't take the necessary care to handle undefined arguments in an intuitive way.
  For example both $\pow(\defdes{x}{\neg(x=x)})=\emptyset$ and
  $\set{\emptyset}\cup(\defdes{x}{\neg(x=x)})=\set\emptyset$ hold.}

\begin{definition}
  The axioms of {\zf}
  are all the instances of the following formulas for every formula $\varphi$ with free variables at
  most $a$, $b$, $c_1$ and $c_2$.
  \begin{enumerate}
    \item Extensionality: $\forall x,y: (\forall a: a \in x \leftrightarrow a \in y) \rightarrow x = y$

    \item Union: $\forall x: \exists y: \forall a: a \in y \leftrightarrow (\exists z\in x: a \in z)$

    \item Power Set: $\forall x: \exists y: \forall z: z \in y \leftrightarrow z \subseteq x$

    \item Infinity (ugly version; see pretty version below):
      \(
        \exists y:        (\bind{\exists}{z\in y}
                             {\bind{\forall}{b}{b\notin z}})
                   \wedge\penalty0\relax
                          (\bind{\forall}{x\in y}
                             {\bind{\exists}{s\in y}
                                {\bind{\forall}{c}
                                   {c\in s \leftrightarrow (c\in x \vee c=x)}}})
      \)

    \item Replacement: $\forall c_1, c_2, x: (\forall a\in x: \exists! b: \varphi) \rightarrow (\exists y: \forall b: b \in y \leftrightarrow \exists a\in x : \varphi)$

    \item Foundation: $\forall x: x = \emptyset \vee (\exists y \in x: \neg \exists b \in x: b \in y)$
  \end{enumerate}
\end{definition}

The axioms are due to Zermelo, except for Replacement which is due to Fraenkel and
Skolem~\cite{ebbinghaus} and Foundation which is due to Von Neumann.
Extensionality asserts that sets are equal iff they contain the same members.
Union and Power Set state that $\Union X$ and $\pow(X)$ are defined if $X$ is defined;
this implies the domain of discourse is closed under $\cup$ and $\pow$.
Infinity states that there exists a set containing $\emptyset$ which is closed under the ordinal
successor operation; from this we can extract the Von Neumann natural numbers $\mathbb{N}$.
Here is a prettier presentation of Infinity that we do not use as the axiom to avoid bootstrap confusion%
\footnote{Provided \emph{some} object exists, Replacement can build $\emptyset$, and then further
  axiom use can build operations like $\set{B,C}$, $\set{B}$, $X \cup Y$, and $\succ{X}$, thus
  ensuring the terms $\emptyset$ and $\succ{x}$ are defined in the pretty version of Infinity.
  We prefer getting that initial object from an axiom over using the FOL assumption that the domain
  of discourse is non-empty.
  The only axiom giving an object for free is Infinity.
  We find it confusing to use Infinity in proving the definedness of subterms of itself, so we
  don't.}%
\relax
:
$$%
  \exists y: \emptyset\in y \wedge (\forall x\in y: \succ{x} \in y)
$$%
The powerful infinite axiom schema Replacement asserts the existence of the range of a function
determined by any formula $\varphi$
where the values of the
variables $a$ and $b$ that make $\varphi$ true have a functional dependency of $b$ on $a$ and where
the domain of the function exists as a set.
Foundation enforces the policy that there are no infinite descending chains of the form $X_0 \ni X_1 \ni
\cdots$.


\begin{lemma}
  The following theorems of {\zf} are often presented as axioms.
  For every formula $\varphi$ such that any free variable must be $a$, the following hold
  in {\zf}:
  \begin{enumerate}
    \item
      Empty Set: $\exists x: \forall b: b\not\in x$
    \item
      Pairing: $\forall a,b: \exists x: \forall c: (c \in x \leftrightarrow (c = a \vee c = b))$
    \item
      Specification: $\forall x: \exists y: \forall a: (a \in y \leftrightarrow (a \in x \wedge \varphi))$
  \end{enumerate}
\end{lemma}


\subsection{Ordered Pairs in \zf}\label{sec:ordered-pairs-zf}

We define the Kuratowski ordered pair $\pair{A}{B}$ and related operations as follows where $a$,
$b$, $p$, and $x$ are not free in $A$, $B$, and $Q$:
$$%
  \begin{array}{@{}l@{}}
      \pair{A}{B} \abbreviates \{\{A\}, \{A,B\}\}
  \\
      A\pl Q \abbreviates (\forall x \in Q: A\in x)
    \qquad
      B\pr Q \abbreviates (\exists! x\in Q: B\in x)
  \\
      A \times B \abbreviates (\riota x: \forall p: p\in x \leftrightarrow (\exists c\in A, d\in B : p = \pair{c}{d}))
  \end{array}
$$%
We call $a$ and $b$ the \emph{first} and \emph{second projections} of $\pair{a}{b}$ respectively.
The first projection of an ordered pair $q$ is in all sets in $q$, whereas the second is only in
one.%
\footnote{This holds even in the case of $\pair{a}{a} =
  \{\{a\},\{a,a\}\} = \{\{a\}\}$.}
\relax
The projection relations $\pl$ and $\pr$ only give meaningful results when the set $Q$ on the right
side of the relation is an ordered pair, i.e., this holds\jbwshortnote{this is not the characteristic property of
  ordered pairs}:
$$
    (\exists c,d: Q = \pair{c}{d})
  \to
    (\forall a, b: (a \pl Q \wedge b \pr Q) \leftrightarrow Q = \pair{a}{b})
$$
Kuratowski ordered pairs are sets and have set members that are distinct from their projections.
In fact, no matter which representation we use, there will always exist some $x$ such that $x \in
\pair{a}{b}$ (for all but at most one ordered pair which can be represented by $\emptyset$).
%
%
If $A$ and $B$ are defined, we can show the cartesian product $A \times B$ is defined using
Replacement nested inside Replacement%
\footnote{%
  The traditional construction of $A \times B$ as
  \(
    \AllSuchThat{p \in \pow(\pow(A \cup B))}
      {\exists c\in A, d\in B : p = \pair{c}{d}}
  \)
  is only needed if the weaker Specification is preferred over Replacement.
  We avoid the traditional construction because it depends on a set representation of ordered pairs
  and thus will not work for {\zfp}.}%
\relax
:
$$%
    A \times B
  = \Union\AllSuchThat{z}{\Exists {c\in A} z = \AllSuchThat{p}{\Exists{d\in B} p = \pair{c}{d}}}
$$%

\section{Extending ZF to ZFP}\label{sec:zfp}

This section introduces Zermelo-Fraenkel Set Theory with Ordered Pairs ({\zfp}), a set theory with
primitive non-set ordered pairs.
{\zfp} axiomatises the membership predicate symbol $\in$ similarly to {\zf}.
The ordered pair projection predicate symbols $\pl$ and $\pr$ are axiomatised in {\zfp}
instead of being abbreviations that use $\in$ as in {\zf}.
Ordered pairs in {\zfp} qualify as urelements because they contain no members via the set membership
relation $\in$, but they are unusual urelements because they can contain arbitrary sets via the
$\pl$ and $\pr$ relations.

\subsection{Definitions and Axioms of ZFP}\label{sec:defin-axioms-zfp}

We use the metavariables $p$, $q$, $P$, and $Q$ where it might help the reader to think `ordered pair',
and the metavariables $s$, $x$, $y$, $z$, $X$, $Y$, and $Z$ where it might help the reader to
think `set'; this convention has no formal status and all FOL variables continue to range over all
objects in the domain of discourse.
We call $b$ a \emph{member} of $x$ iff $b \in x$.
We call $b$ a \emph{projection} of $q$ iff $b \pl q$ or $b \pr q$.
An \emph{ordered pair} is any object with a projection, and a \emph{set} is any object that is not
an ordered pair.
We use the following abbreviations where $b$ is not free in $Q$ and $X$ and $q$ is not
free in $A$ and $B$:
$$%
  \begin{array}{@{}l@{\ \abbreviates\ }l@{\qquad}l@{\ \abbreviates\ }l@{}}
      \Pair(Q)         & \exists b: b\pl Q
    &
      \Set(X)          & \neg \Pair(X)
  \\
      \allP p: \varphi & \forall p: \Pair(p) \rightarrow \varphi
    &
      \allS x: \varphi & \forall x: \Set(x) \rightarrow \varphi
  \\
      \exP p: \varphi  & \exists p: \Pair(p) \wedge \varphi
    &
      \exS x: \varphi  & \exists x: \Set(x) \wedge \varphi
  \\
      \unP p: \varphi  & \riota p: \Pair(p) \wedge \varphi
    &
      \unS x: \varphi  &\riota x: \Set(x) \wedge \varphi
  \\
    (A,B) & (\riota q: A \pl q \wedge B \pr q)
  \end{array}
$$%
We reuse the text of the abbreviation definitions for {\zf} for $\set{A,B}$,
$X\cup Y$, $\set{A}$, and $\set{A_1,\ldots,A_n}$ where $n\geq 3$.
We redefine the following abbreviations a bit differently for {\zfp},
where $a$, $b$, $c$, $p$, $x$, $y$, and $z$ are not free in $A$, $B$, $X$ and $Y$:
$$%
  \begin{array}{@{}l@{\ }l@{}}
    X \subseteq Y &\abbreviates\ \Set(X) \wedge \Set(Y) \wedge (\bind{\forall}{c\in X}{c\in Y})
  \\
    \Union X &\abbreviates\ (\unS y: \forall a: a\in y \leftrightarrow \exists z\in X: a\in z)
  \\
    \pow(X) &\abbreviates\ (\bind{\unS}{y}{\bind{\forall}{z}{z \in y \leftrightarrow z \subseteq X}})
  \\
    \emptyset &\abbreviates\ (\unS x: \forall a: a\notin x)
  \\
                 \AllSuchThat{b \in X}{\varphi}
    &\abbreviates\ (\bind{\unS}{y}{\forall b: b \in y \leftrightarrow (b \in X \wedge \varphi)})
  \\
      A \times B
    &\abbreviates\
      (\riota x: \forall p: p\in x \leftrightarrow (\exists c\in A, d\in B : p = ({c},{d})))
  \end{array}
$$%
These abbreviations are defined if their arguments are defined due to the axioms.

\jbwlongnote{In {\zfp} $\pow((\emptyset,\emptyset))=\emptyset$ holds.
  This is probably not what we want for a practical version of the system.}


\begin{definition}
  The axioms of {\zfp} are all the instances of the following formulas for every formula
  $\varphi$ with free variables at most $a$, $b$, $c_1$, $c_2$.
  \begin{itemize}
  \item \textbf{Sets}:
    \begin{enumerate}[label=\textbf{S\arabic*.}, ref=S\arabic*]
    \item Set Extensionality:
    $\allS x,y: (\forall a: a\in x \leftrightarrow a\in y) \rightarrow x=y$ \label{set_ext}

    \item Union:
    $\allS x: \exists y: \forall a: a\in y \leftrightarrow (\exists z\in x : a\in z)$ \label{union}

    \item Power Set:
    $\allS x:\exists y:\forall z: z\in y \leftrightarrow z\subseteq x$ \label{power}

    \item Infinity (ugly version):
      \(
        \exists y:        (\bind{\exS}{z\in y}
                             {\bind{\forall}{b}{b\notin z}})
                   \wedge\penalty0\relax
                          (\bind{\forall}{x\in y}
                             {\bind{\exists}{s\in y}
                                {\bind{\forall}{c}
                                   {c\in s \leftrightarrow (c\in x \vee c=x)}}})
      \).%
      \label{inf}%
      \jbwlongnote{The set witnessing Infinity can contain ordered pairs and for each such ordered
        pair $Q$ will contain also $\set{Q}$.}

    \item Replacement:\\
    $\forall c_1, c_2, x: (\forall a\in x : \exists!b: \varphi) \rightarrow (\exS y: \forall b: b\in y \leftrightarrow \exists a\in x: \varphi)$
    \label{rep}

    \item Foundation:
    $\allS x: x = \emptyset \vee (\exists a\in x:\neg \exists b\in x: b\pl a\vee b\pr a \vee b\in a)$
    \label{found}
    \end{enumerate}
  \item \textbf{Ordered Pairs:}
    \begin{enumerate}[label=\textbf{P\arabic*.}, ref=P\arabic*]
    \item Ordered Pair Emptiness:
    $\allP p: \forall a: a\notin p$
    \label{ax:pair-empty}

    \item Ordered Pair Formation:
    $\forall a,b: \exists p: a\pl p \wedge b\pr p$
    \label{pair_form}

    \item Projection Both-Or-Neither:
    $\forall p: (\exists a: a\pl p) \leftrightarrow (\exists b: b\pr p)$
    \label{proj_ex}

    \item Projection Uniqueness:
    $\allP p: (\exists! a: a\pl p) \wedge (\exists! b: b\pr p)$
    \label{ax:pair-uniq}

    \item Ordered Pair Extensionality:
    \\ $\allP p,q: (\forall a: (a\pl p \leftrightarrow a\pl q) \wedge (a\pr p \leftrightarrow a\pr q)) \rightarrow p=q$
    \label{pair_ext}

    \end{enumerate}
  \end{itemize}
\end{definition}

\begin{lemma}
  \label{lem:zfp-derived-properties}
  For every formula $\varphi$ such that any free variable must be $a$, the following hold in {\zfp}:
  \begin{enumerate}
  \item Unordered/Set Pairing:
    $\forall a,b: \exists x: \forall c: c\in x \leftrightarrow (c=a \vee c=b)$ \label{set_pair}
    \jbwshortnote{Proven that pairing and specification are derivable?}

  \item Specification:
    $\allS x: \exS y: \forall a: a\in y \leftrightarrow (a\in x \wedge \varphi))$ \label{spec}

  \item
    \label{cart_prod}
    Cartesian Product Existence:\\
    \(
      \allS x,y: \exS z:\forall p:
          p \in z
        \leftrightarrow
          (\exists a \in x, b\in y : a\pl p \wedge b\pr p)
    \)
  \end{enumerate}
\end{lemma}

For \autoref{lem:zfp-derived-properties}~(\ref{cart_prod}), note that the cartesian product
$A\times{B}$ can be built in {\zfp} using the same construction given for {\zf} in
\autoref{sec:ordered-pairs-zf}, which does not depend on any set representation of ordered pairs.

\subsection{Discussion}

\subsubsection{Axioms for Sets.}

Each {\zf} axiom was transformed to make a {\zfp} axiom.
First, because we use abbreviations for more readable axioms, those used in axioms needed to be
modified for {\zfp}.
The definition of $\subseteq$ (used in Power Set) was changed to ensure an ordered pair is neither a
subset nor has a subset%
.
The definition of $\emptyset$ (used in Foundation) was changed to ensure a defined result.

Second, some occurrences of $(\bind{\forall}{b}{\psi})$ and $(\bind{\exists}{b}{\psi})$ needed to
enforce that $\psi$ can be true only when $b$ stands for a set.
Where needed, such occurrences were changed to $(\bind{\allS}{b}{\psi})$ respectively
$(\bind{\exS}{b}{\psi})$.
Each quantifier needed individual consideration.
If the sethood of $b$ was already enforced by $\psi$ only being true when $b$ has at least 1 set
member, there was no need for a change but a change might also clarify the axiom.
If the truth of $\psi$ was unaffected by any set members of $b$, there was no need for a change and
this generally indicated that a change would go against the axiom's intention.
We needed to understand the axiom's \emph{intention} and \emph{expected usage} because it was not
written to specify where it is expected that `$X$ is a set' (because this always holds in {\zf}).

Finally, Foundation was extended to enforce a policy of no infinite descending chains through not
just $\in$ but also $\pl$ and $\pr$, so that {\zf} proofs using Kuratowski ordered pairs (having no such
chains) would continue to work in {\zfp}.

Consider the example of Power Set
which states that for any set $X$ there exists a set $Y$ containing all of the subsets of $X$ and
nothing else, i.e., $\pow(X)$:
$$
  \allS x: \exists y: \forall z: (z\in y \leftrightarrow z \subseteq x)
$$
We could have left $\allS x$ as $\forall x$, because when $x$ is an ordered pair it would act like
$\emptyset$ and this would only add another reason that $\pow(\emptyset)$ exists.
However, we thought this would be obscure.
It would not hurt to change $\exists y$ to $\exS y$ but there is no need to do so because the body
forces $y$ to contain a set member and hence rejects $y$ being an ordered pair.
We did not change $\forall z$ to $\allS z$ because this would allow $y$ to contain extra junk
ordered pairs that proofs expecting to get $\pow(x)$ would have to do extra work using Replacement
to filter out.

\subsubsection{Axioms for Ordered Pairs.}

The {\zfp} axioms for ordered pairs specify the abstract properties of ordered pairs via the
relations $\pl$ and $\pr$.
These ordered pairs have no `type' restrictions, i.e., each pair projection can be either a set or an ordered pair.
Ordered Pair Emptiness~(\ref{ax:pair-empty}) ensures that no object has both a projection (ordered pairs
only) and a set member (sets only).
Ordered Pair Formation~(\ref{pair_form}) ensures that for every two objects $b$ and $c$ there exists
an ordered pair with $b$ as first projection and $c$ as second.
Projection Both-Or-Neither~(\ref{proj_ex}) ensures that every object either has no projections
(sets) or both projections (ordered pairs).
Projection Uniqueness~(\ref{ax:pair-uniq}) ensures each ordered pair has exactly one first projection and
one second projection.
Ordered Pair Extensionality~(\ref{pair_ext}) ensures that for every choice of first and second
projections, there is exactly one ordered pair.

\subsubsection{Comparing the Objects and Theorems of {\zf} and {\zfp}.}

A set is \emph{pure} iff all its members are pure sets.
Each {\zf} object is a pure set and is also a pure set of {\zfp}, but {\zfp} has additional impure
sets which have members that are primitive ordered pairs or impure sets, and {\zfp} also has
primitive ordered pairs.
The set membership relation $\in$ of {\zf} is the restriction of the relation $\in$ of {\zfp} to
pure sets.
Let $\mathsf{Pure}(x)$ be a formula (implemented with transfinite recursion) that holds in {\zfp}
when $x$ is a pure set.
For every {\zf} formula $\varphi$, let $\mathsf{PRestrict}(\varphi)$ be the {\zfp} formula obtained
from $\varphi$ by changing each subformula $(\bind{\forall}{x}{\psi})$ to
$(\bind{\forall}{x}{\mathsf{Pure}(x)\to\psi})$.
Then $\varphi$ is a {\zf} theorem iff $\mathsf{PRestrict}(\varphi)$ is a {\zfp} theorem.
If one wants to go the other direction and take a {\zfp} formula $\psi$ and find a {\zf} formula $\psi'$ that `does the same
thing', one must represent as {\zf} sets both (1)~the primitive ordered pairs \emph{and} (2)~the
sets of {\zfp}, and then one must either prevent or somehow manage the possible confusion between
the representations of (1) and~(2).
\Autoref{sec:interpreting-zfp-zf} is an example of doing this rigorously.

\subsubsection{Design Alternatives.}

We considered having the projections $\pl$ and $\pr$ be unary FOL function symbols, but this would
require the term ${\pl}(x)$ to denote an object within the domain of discourse for every set $x$, so
we avoided this.
We considered having the pairing operator $(\cdot, \cdot)$ be a binary FOL function symbol.
Using a binary function symbol would mean the graph model would have hyperedges
(i.e., connecting 3 or more nodes) which is more difficult to think about.
Because we used two separate binary predicate symbols, one for each projection, we get a fairly
standard-looking directed-graph
model with ordinary edges.
If we used a binary FOL function symbol $(\cdot, \cdot)$ for pairing, we could replace our axioms
\ref{pair_form}, \ref{proj_ex}, \ref{ax:pair-uniq}, and \ref{pair_ext} by the characteristic
property of ordered pairs:
$$%
  \bind{\forall}{a,b,c,d}{(a,b) = (c,d) \rightarrow (a=b \wedge c=d)}
$$%
Our axioms can be seen as the result of applying a function-symbol-elimination transformation to
this alternative.

Very early on, we considered simply using {\zf}'s axioms as they are, adding a binary pairing
function symbol, and adding the characteristic property of ordered pairs as an axiom.
In this theory, formulas such as $\set{b} \in \pair{b}{c}$ would be independent, because the
representation of ordered pairs would be unknown (and need not even be definable in ZF), so some
`junk theorems' would no longer hold.
We avoided this alternative for many reasons.
First, Extensionality would force all but one ordered pair (which could be $\emptyset$) to have set
members, so there would be `junk theorems' such as
\(
  (a,b) \ne (c,d) \rightarrow \bind{\exists}{e}{e \in (a,b) \leftrightarrow e \notin (c,d)}
\).
\relax
Second, we could not see how to do transfinite induction and recursion.
Third, genuine non-sets make it easier to talk about the distinction between sets and conceptually
non-set objects, e.g., to students.
Fourth, we hope our approach might help a weak form of `type checking', where a prover might more
quickly solve or disprove subgoals, and if a user mistakenly requires a non-set to have a set
member, this might be detected earlier and result in a more understandable failure message.
Some further reasons are discussed in \autoref{sec:introduction}.

\section{A Model of ZFP}\label{sec:zfp-model}

We define within {\zf} a model for {\zfp}, i.e., an interpretation of the domain and predicate
symbols of {\zfp}.
A translation from a ZFP formula $\psi$ to a ZF formula $\psi^*$ is defined to interpret ZFP
formulas
in the model.
Terms and formulas in this section belong to {\zf} except for the arguments of $(\:\cdot\:)^*$.
All axioms of {\zfp} hold under this translation, which implies that if {\zf} is consistent,
so is {\zfp} \cite{enderton}.
That each axiom's translation holds has been checked in Isabelle/{\zf}.

\subsection{The Cumulative Hierarchy W}

Like the Von Neumann universe {\V} used as the domain of a model of {\zf}, our domain
{\W} is a set hierarchy indexed by ordinal numbers.

An \emph{ordinal} is a \emph{transitive} set that is totally ordered by $\in$, which we specify
formally by
\(
    \Ord(x)
  \abbreviates
      (\bind{\forall}{y\in x}{y\subseteq x})
    \wedge
      (\bind{\forall}{y,z\in x}{y=z \vee y\in z \vee z\in y})
\)%
\relax
\jbwlongnote{The definition of $\Ord$ needs a condition $X=X$ to be safely used on definite descriptions.}%
.
Let $\alpha$ and $\beta$ range over ordinals.
Let $0\abbreviates\emptyset$, $1\abbreviates0^+$, $2\abbreviates1^+$, and so on.
Ordinal $\beta$ is a \emph{successor} ordinal
iff $\beta=\alpha^+$ for some $\alpha$.
Ordinal $\beta$ is a \emph{limit} ordinal iff $\beta$ is neither $0$ nor a successor ordinal.
Let $\lambda$ range over limit ordinals.
Let $(x<y)\abbreviates(x\in y\wedge\Ord(y))$ and define related symbols (e.g., $\le$) as usual.


Any model of ZFP must have some way of distinguishing between the  objects in its domain representing ZFP
sets, and those that represent ZFP pairs, i.e., {\zfp} needs a domain split into two disjoint subdomains.
We model this in {\zf} using Kuratowski ordered pairs and cartesian products to tag all domain
objects with $0$ (`set') or $1$ (`ordered pair').

\begin{definition}
  For ordinal $\alpha$, define the set\/ $W_\alpha$ via transfinite recursion thus:
  $$
    W_0 = \emptyset
  ,\qquad
    W_{\beta^+} = (\{0\} \times \pow(W_\beta)) \cup (\{1\} \times (W_\beta)^2)
  ,\qquad
    W_{\lambda} =  \bigcup_{\beta \in \lambda} W_\beta
  $$
\end{definition}

Starting from $\emptyset$, each successor tier $W_{\beta^+}$ is built by taking the disjoint union
of the power set and cartesian square of the previous tier.
Each limit tier $W_\lambda$ is the union of all preceding tiers.
The use of disjoint union to build each successor tier $W_{\beta^+}$ gives a set-theoretic universe
split into two.
Although our disjoint union uses Kuratowski pairs with 0 and 1 tags, we could use instead any two
definable injective operators from a large enough class (e.g., the universe) to disjoint classes that raise
rank by at most a constant.

Let \W{} be the proper class such that $x\in\W$ iff $x\in W_\alpha$ for some $\alpha$.
We use a bold upright serif font to emphasize that \W~is not a {\zf} set.%
\footnote{\W~is a
  mathematical object in some other set theories.}
By the transfinite recursion theorem, given $x$ there is a definite description $\mathsf{W}(x)$
that evaluates to $W_\alpha$ when $x$ evaluates to $\alpha$.%
\footnote{A nested definite description is used
  that specifies the function $f$ such that $f(\beta)=W_\beta$ for $\beta\leq\alpha$, i.e., $f$ is
  an initial prefix of the hierarchy.
  Then $f(\alpha)$ is returned.}
We express $X$ belonging to \W{} as follows:

\begin{definition}
  $\mathcal{H}(X) \abbreviates (\bind{\exists}{y}{\Ord(y) \wedge X\in \mathsf{W}(y)})$.
\end{definition}

Let an \emph{m-object} be any member of {\W} (i.e., a {\zf} set $x$ such that $\mathcal{H}(x)$
holds), an \emph{m-set} be any m-object of the form $\pair{0}{x}$, and an \emph{m-pair} be any
m-object of the form $\pair{1}{x}$.
The following result says every m-object $x$ is either an m-set or an m-pair, and tells where in
the hierarchy the contents of $x$ are.

\begin{lemma}\label{thm:zset-or-zpair}
  Suppose $\mathcal{H}(x)$, so that $x\in W_\alpha$.
  Then for some $\beta<\alpha$ either:
  $$
      x = \pair{0}{x'} \text{ where $x'\subseteq W_\beta$,}
   \quad\text{or}\quad
      x = \pair{1}{\pair{a}{b}} \text{ where $a,b \in W_\beta$.}
  $$
\end{lemma}

It holds that {\W} is a cumulative hierarchy:


\begin{lemma}\label{lem:W-subset}
        If $\alpha \leq \beta$, then $W_\alpha \subseteq W_\beta$.
\end{lemma}

\subsection{Interpreting ZFP in ZF}
\label{sec:interpreting-zfp-zf}

As explained above, we interpret the sets and ordered pairs of {\zfp} as the members of {\W}.
\Autoref{thm:zset-or-zpair} says any m-object is an ordered pair whose left projection is
an integer which decides its `type' and whose right projection is either a set or an ordered
pair.
We define our interpretations of {\zfp}'s predicate symbols:

\begin{definition}\label{zf_rel}
  Let $\zin$, $\zpl$, and $\zpr$ be defined by these abbreviations:
  $$
    \begin{array}{@{}l@{\ \abbreviates\ }l@{}}
       a \zin x & (\bind{\exists}{y}{x = \pair{0}{y} \wedge a \in y})
    \\
       a \zpl p & (\bind{\exists}{u,v}{p = \pair{1}{\pair{u}{v}} \wedge a = u})
    \\
       a \zpr p & (\bind{\exists}{u,v}{p = \pair{1}{\pair{u}{v}} \wedge a = v})
    \end{array}
  $$
\end{definition}

{\W} is downward closed under these three relations\cmdshortnote{\emph{downwards closed} is a
  property of partially ordered sets, and here {\W} isn't even a math object, so unsure how to state
  this}.
That is:

\begin{lemma}\label{prev_tier_rel}
  Suppose $\mathcal{H}(x)$, i.e., $x\in W_\alpha$ for some $\alpha$.
  Suppose $a \zin x$, $a \zpl x$, or $a \zpr x$ for some $a$.
  Then $a\in W_\beta$ for some $\beta < \alpha$, and thus $\mathcal{H}(a)$.
  %
\end{lemma}

To interpret a {\zfp} formula $\varphi$ in {\zf}, we must show the formula holds when
quantification is restricted to the domain {\W}, and the predicate symbols are replaced by the
interpretations defined above.

\begin{definition}\label{def:translation}
  Let $\varphi$ be a {\zfp} formula.
  Define $\varphi^*$ recursively as follows:%
  $$%
    \begin{array}{@{}l@{\ \abbreviates\ }l@{\qquad}l@{\ \abbreviates\ }l@{}}
        (X \in Y)^* &  (X^*) \zin (Y^*)
      &
        (\varphi \rightarrow \psi)^* & (\varphi^*) \rightarrow (\psi^*)
    \\
        (X \pl Y)^* & (X^*) \zpl (Y^*)
      &
        (\neg \varphi)^* & \neg (\varphi^*)
    \\
        (X \pr Y)^* & (X^*) \zpr (Y^*)
      &
        (\forall x: \varphi)^* & (\bind{\forall}{x}{\mathcal{H}(x) \rightarrow (\varphi^*)})
    \\
        x^* & x
      &
        (\defdes{x}{\varphi})^* & (\defdes{x}{\mathcal{H}(x) \wedge (\varphi^*))}
    \end{array}
  $$%
\end{definition}

\begin{lemma}\label{lem:translation-existential}
  \(
      (\bind{\exists}{x}{\varphi})^*
    \leftrightarrow
      (\bind{\exists}{x}{\mathcal{H}(x)\wedge (\varphi^*)})
  \).
\end{lemma}

Because the translation $(\:\cdot\:)^*$ inserts quite a lot of extra structure, a {\zfp} user
wanting to understand ``the {\zf} formula corresponding to the {\zfp} formula $\psi$'' might be
tempted to instead translate {\zfp}'s $\in$ directly to {\zf}'s $\in$ and {\zfp}'s $\pl$ and $\pr$
to the {\zf} abbreviations for $\pl$ and $\pr$ defined in \autoref{sec:ordered-pairs-zf}.
However, as discussed in \autoref{sec:issue-representation}, the user then would need to carefully
prove that no problems arise from the coincidences where a {\zfp} set $x$ and a {\zfp} primitive
ordered pair $p$ would be represented by the same {\zf} set $y$.

Observe that the {\zfp} abbreviations $\Set$ and $\Pair$ from \autoref{sec:defin-axioms-zfp} that
act like unary predicates are interpreted in {\zf} as follows:
$$%
    \Pair(x)^* \abbreviates (\bind{\exists}{a}{\mathcal{H}(a) \wedge a \zpl x})
  \qquad\qquad
    \Set(x)^* \abbreviates \neg(\Pair(x)^*)
$$%
These predicates are clearly meaningful within the model because:
\begin{lemma}\label{lem:u_set}
  Suppose that $\mathcal{H}(x)$, then we have that:
  $$%
      \Pair(x)^* \leftrightarrow (\bind{\exists}{a,b}{x = \pair{1}{\pair{a}{b}}})
    \qquad\qquad
      \Set(x)^*  \leftrightarrow (\bind{\exists}{y}{x = \pair{0}{y}})
  $$%
\end{lemma}

Now we reach our main result, which implies {\zfp} is consistent if {\zf} is~\cite{enderton}:

\begin{theorem}
  For each {\zfp} axiom $\varphi$, the translation $\varphi^*$ holds in {\zf}.
\end{theorem}

The proof of this theorem simply observes the conjunction of a number of lemmas, each of which shows
for a {\zfp} axiom $\phi$ that $\phi^*$ holds in {\zf}.
Most of these lemmas are straightforward.
Here we show a representative example:

\begin{lemma}
  The translation of {\zfp}'s Power Set axiom holds in {\zf}.
\end{lemma}

\begin{proof}
    First, we find the translation using \autoref{def:translation} and \autoref{lem:translation-existential}:
    $$
      \forall x:
          \mathcal{H}(x)
        \rightarrow
          (
             {\Set(x)}^*
           \rightarrow
             (
              \exists y:
                  \mathcal{H}(y)
                \wedge
                  \forall z:
                      \mathcal{H}(z)
                    \rightarrow
                      (
                         z \zin y
                       \leftrightarrow
                         ((z \subseteq x)^*)
                      )
             )
          )
    $$
    Let $x$ be such that $\mathcal{H}(x)$, and suppose ${\Set(x)}^*$.
    By~\autoref{lem:u_set}, $x = \pair{0}{x'}$ for some set $x'$.
    Let $y = \pair{0}{y'}$ where $y' = \{0\} \times \mathcal{P}(x')$ be our candidate for the power set.
    We must show that $y$ has the property $\forall z: \mathcal{H}(z) \rightarrow (z \zin y \leftrightarrow (z \subseteq x)^*)$, and also that $y$ is indeed a member of {\W}.
    Fix $z$ and assume $\mathcal{H}(z)$, then:
    \begin{align*}
        z \zin y
        &\leftrightarrow z \in y'
        \tag*{by def of $y$ and $\zin$}
    \\
        &\leftrightarrow z \in \{0\} \times \mathcal{P}(x')
        \tag*{by def of $y'$}
    \\
        &\leftrightarrow \exists z': z = \pair{0}{z'} \wedge z' \subseteq x'
        \tag*{by def of $\times$ and $\mathcal{P}$}
    \\
        &\leftrightarrow {\Set(z)}^* \wedge (\forall a: a \zin z \rightarrow a \zin x)
        \tag*{since $z = \pair{0}{z'}, z' \subseteq x'$}
    \\
        &\leftrightarrow {\Set(z)}^* \wedge {\Set(x)}^* \wedge (\forall a: a \zin z \rightarrow a \zin x)
        \tag*{since $\mathcal{H}(x)$, $x = \pair{0}{x'}$}
    \\
      & \leftrightarrow (z \subseteq x)^*
        \tag*{because $\mathcal{H}(z)$}
    \end{align*}
    It now remains to show that $\mathcal{H}(y)$.
    From $\mathcal{H}(x)$, we have that $x\in W_{\alpha}$ for some ordinal $\alpha$.
    By \autoref{lem:W-subset}, $x\in W_{\alpha^+}$, and by \autoref{thm:zset-or-zpair}, $x' \subseteq W_\alpha$. Then:
    \begin{align*}
        x'\subseteq W_\alpha
        &\rightarrow \mathcal{P}(x') \subseteq \mathcal{P}(W_\alpha)\\
        &\rightarrow \{0\}\times \mathcal{P}(x') \subseteq \{0\}\times \mathcal{P}(W_\alpha) \\
        &\rightarrow y' \subseteq \{0\}\times \mathcal{P}(W_\alpha)
        \tag*{by def of $y'$}\\
        &\rightarrow y' \subseteq W_{\alpha^+}
        \tag*{because $\{0\}\times \mathcal{P}(W_\alpha) \subseteq W_{\alpha^+}$} \\
        &\rightarrow y\in W_{\alpha^{++}}
        \tag*{by def of $y = \pair{0}{y'}$}\\
        &\rightarrow \mathcal{H}(y)
        \tag*{by def of $\mathcal{H}$}
        \end{align*} \qed
    \end{proof}

\section{Conclusion}\label{sec:future}

\subsection{Summary of Contributions}

\subsubsection{Presenting {\zf} Set Theory using Definite Descriptions.}

In \autoref{sec:fol-zf}, we give a formal presentation of {\zf} that accounts for the technical
details, whilst also defining notation for widely used operations.
Although correct formal definitions of this notation can be found in computer implementations of set
theory, we have not seen definite descriptions used for this in published articles.
Definite descriptions allow defining terms in a compact and readable way without needing to add FOL
function symbols, extend the model, or otherwise appeal to the meta-level.
We show precisely how Kuratowski pairs and their operations are defined
and highlight issues arising from their set representations.

\subsubsection{Axiomatizing {\zfp}.}

Motivated by issues with the set representation in pure {\zf} set theory of conceptually non-set
objects, in \autoref{sec:zfp} we introduce Zermelo-Fraenkel Set Theory with Ordered Pairs, which
extends {\zf} with predicate symbols $\pl$ and $\pr$ and axioms to implement primitive non-set
ordered pairs.
{\zfp} is akin to some alternative set theories that use urelements as genuine non-set objects in
the domain, with the difference that ZFP's urelements have meaningful internal structure endowed by the
axiomatisation of $\pl$ and $\pr$.
The design of {\zfp} is deliberately similar to that of {\zf}, so that we can better understand the
relationship between the two theories.
We axiomatize {\zfp}, and discuss how the axioms of {\zf} were
modified to yield the axioms of {\zfp}.
As a result, we gain a set theory with two types of individuals,
both of which have a notion of `container',
which is unusual as urelements are usually structureless.
The primitive ordered pairs of {\zfp} are unlike those typical of set theory, as they are free from
any notion of representation.

\subsubsection{Showing {\zfp} Consistent.}

In \autoref{sec:zfp-model}, we construct a transfinite hierarchy to be the domain of a model for
{\zfp} and we define relations on this domain to be interpretations for $\in$, $\pl$, and $\pr$.
We show that the resulting structure satisfies the axioms of {\zfp}, i.e., it is a model for
{\zfp}.
As a result, we show {\zfp} is consistent if {\zf} is.



\subsection{Future Work}

\subsubsection{Model Theoretic Status of {\zf} and {\zfp}.}

Axiomatisations of both {\zf} and {\zfp} are given within this paper, and we are aware that the sets
of {\zfp} behave in a similar fashion to those in {\zf}.
We suggest employing model-theoretic techniques to give a more detailed formal account of the relationship between
the formulas of both theories, as well as the models.

\subsubsection{Implementing {\zfp}.}

Preliminary experiments have taken place in implementing {\zfp} as an object logic for Isabelle.
Further work on this will allow comparing mathematics formalised in {\zf} and in {\zfp}, and thus
allow comparing the expressivity, and automatability of both theories.
Moreover, there is already a large library of mathematics formalised in Isabelle/ZF.
Once the formal relationship between {\zf} and {\zfp} has been established,
we will attempt to translate mathematics between both bases.

\subsubsection{Towards Abstract Data Types in Set Theory.}

In this paper we identified a role performed by some sets in {\zf}, namely the role of being an ordered
pair for some representation (e.g., Kuratowski),
together with the FOL abbreviations for their relations.
We axiomatised a new set theory in which this role can be performed by non-set objects,
yet maintain the same existence conditions and abstract behaviour of this role.
We will attempt to abstract and adapt this method, to yield set theories in which the members of
mathematical structures can be genuine non-sets dedicated to their role.
We believe such a framework could be helpful when using
set theory to formalise mathematics.

\bibliography{zfp-paper}
\bibliographystyle{plain}
\end{document}